\documentclass[sigconf]{aamas} 
\setcopyright{none}
\renewcommand\footnotetextcopyrightpermission[1]{}
\usepackage{balance} 

\usepackage{booktabs} 
\usepackage[ruled]{algorithm2e}
\usepackage{graphicx}
\usepackage{float} 
\usepackage{subfigure}
\usepackage{enumitem}
\usepackage{bbm}
\usepackage[utf8]{inputenc} 
\usepackage[T1]{fontenc}    
\usepackage{hyperref}       
\usepackage{url}            
\usepackage{amsfonts}       
\usepackage{nicefrac}       
\usepackage{microtype}      
\usepackage{xcolor}         

\usepackage{amsmath}
\usepackage{amsthm}
\usepackage{multirow}
\usepackage{pifont}
\usepackage{threeparttable}

\newcommand{\argmax}{\operatorname{argmax}}

\newcommand{\indicator}[1]{\mathcal{I}\left(#1\right)}
\newcommand{\M}{\textit{OPT-Pay}}
\newcommand{\Single}{\textit{OPT-Standard}}

\newcommand{\W}{\textit{WEL}}
\newcommand{\T}{\textit{OPT-Terminate}}
\newcommand{\Linear}{\textit{OPT-Linear}}
\newtheorem{remark}{Remark}

\newtheorem{claims}{Claim}
\newtheorem{main result}{Main Result}

\makeatletter
\gdef\@copyrightpermission{
  \begin{minipage}{0.2\columnwidth}
   \href{https://creativecommons.org/licenses/by/4.0/}{\includegraphics[width=0.90\textwidth]{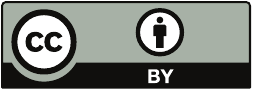}}
  \end{minipage}\hfill
  \begin{minipage}{0.8\columnwidth}
   \href{https://creativecommons.org/licenses/by/4.0/}{This work is licensed under a Creative Commons Attribution International 4.0 License.}
  \end{minipage}
  \vspace{5pt}
}
\makeatother

\title[AAMAS-2026 Formatting Instructions]{The Power of Information for Intermediate States in Contract Design}

    \author{Yirui Zhang}
\affiliation{
  \institution{Tsinghua University}
  \city{Beijing}
  \country{China}}
\email{zhangyr22@mails.tsinghua.edu.cn}

\author{Zhixuan Fang}
\affiliation{
 \institution{Tsinghua University}
 \thanks{Corresponding author: Zhixuan Fang at \emph{zfang@mail.tsinghua.edu.cn}. \\
 This work is supported by Tsinghua University Dushi Program.}
  \city{Beijing}
  \country{China}}
  \affiliation{
 \institution{Shanghai Qi Zhi Institute}
  \city{Shanghai}
  \country{China}}
\email{zfang@mail.tsinghua.edu.cn}

\begin{abstract}
In the conventional principal-agent problem, a principal delegates a task to an agent and formulates a contract to incentivize the agent’s actions on behalf of the principal. However, this framework overlooks the information that is possibly available during the delegation process in some scenarios.  To address this limitation, we propose a novel model that incorporates multiple intermediate states to capture such information revealed during the delegation. Furthermore, to evaluate the impact of the information embedded in these intermediate states, we introduce two distinct contracts: the pay-halfway contract, which provides payments based not only on final outcomes but also on intermediate states, and the terminate-halfway contract, which allows the principal to terminate the delegation process upon encountering undesirable intermediate states.
This leads to the question of whether and how these contract types can leverage intermediate-state information? In particular, we ask: Can these contract types outperform standard contracts, and if so, when and to what extent?  We answer the first question affirmatively and provide several important insights regarding the second, shedding light on the circumstances in which intermediate-state-aware contracts yield substantial advantages. 
\end{abstract}

\keywords{Contract Design}

\newcommand{\BibTeX}{\rm B\kern-.05em{\sc i\kern-.025em b}\kern-.08em\TeX}

\begin{document}


 \pagestyle{fancy}
 \fancyhead{}


\maketitle 


\section{Introduction}
Contract theory is a fundamental area of economics that examines how contractual arrangements are formulated in the presence of information asymmetry.  Prior research has typically utilized the principal-agent framework to study the dynamic between the involved parties, where the principal (employer) often lacks complete information about the agent's (employee's) actions or effort levels. This uncertainty necessitates the creation of incentive-compatible contracts that motivate the agent to act in the principal's best interest while considering the agent's self-interest.

In the classic setup, the principal delegates decision-making authority to the agent and can only observe the final outcome, not the actions taken by the agent. This leads to a potential moral hazard, as the agent may not act in the best interest of the principal. To mitigate this problem, the principal designs incentive-compatible contracts that link compensation to observable outcomes in order to encourage desirable behavior.  The agent, seeking to maximize her own benefits, will then take actions based on the promised compensations. Much of the existing literature in this field (\cite{bolton2004contract,dutting2024combinatorial,guruganesh2023power}) is dedicated to determining the optimal contract arrangement from the principal's perspective.

However, real-world delegation often allows the principal to observe some information during the process. For instance, when assigning a project to an employee, the employer may receive mid-term progress reports, enabling them to monitor the project's advancement. Similarly, in supply chain scenarios, suppliers may provide periodic updates on the production process. Such intermediate information is not captured in the standard principal-agent framework.
To address this limitation, we propose a new model within contract theory: the two-stage delegation process. This model incorporates intermediate states to capture information revealed during the delegation. Specifically, the process consists of two stages: 
(i) The agent takes an initial action and transitions to a specific intermediate state, which embeds the information about the delegation process.
(ii) At the intermediate state, the agent takes further action, resulting in an outcome. Importantly, the principal can observe both the intermediate state and the final outcome.

This extended framework provides a more accurate representation of the dynamic nature of real-world contractual environments. Many practical settings fit naturally into our framework, including procurement, crowdsourcing, and adaptive task assignment. For example, in procurement, the evolving stages of a project can be interpreted as intermediate information revealed progressively over the course of the delegation process. Consequently, within this two-stage model that incorporates intermediate states, it becomes both interesting and necessary to analyze how such information revelation influences contractual relationships and outcomes.

\subsection{Main Results}\label{sec_contribution}
In this section, we summarize our main results, focusing on the central question: whether and how information embedded in intermediate states benefits the principal. 

To explore this, we introduce two conceptually distinct contractual mechanisms: the pay-halfway contract and the terminate-halfway contract. Unlike standard contracts, which provide compensation solely based on the final outcome, the pay-halfway contract allows for compensation upon the realization of specific intermediate states, while the terminate-halfway contract permits early termination of the delegation process when undesirable intermediate states are observed. We analyze and compare the performance of these contracts against the standard outcome-based contract. 

The information embedded in intermediate states comprises two key types: (I) the intrinsic quality of the state, meaning how likely final actions from this state lead to valuable outcomes; and (II) the agent’s initial action choice. 

To examine which type(s) of information benefit the principal, and how the proposed contracts leverage them, we consider three specialized two-stage delegation processes that differ in the informational content of their intermediate states:

\begin{itemize}
    \item In the tree two-stage process, each final outcome is uniquely linked to a single intermediate state, rendering the intermediate state uninformative beyond the final outcome. (No additional information)
    \item In the stochastic first-stage process, there is only one initial action; although intermediate states do not reveal any information about the agent’s initial choice, they carry predictive value regarding the likelihood of favorable final outcomes. (Type I information only)
    \item In the deterministic first-stage process, each intermediate state corresponds to a unique initial action, thereby containing information both about the state itself and fully revealing the agent’s initial choice. (Type I and II information)
\end{itemize}

Our results demonstrate the following:

\begin{itemize}
    \item In the tree two-stage process, both the pay-halfway and terminate-halfway contracts perform identically to the standard contract due to the lack of informative intermediate states.
    \item In the stochastic first-stage process, the terminate-halfway contract can significantly outperform the standard contract, while the pay-halfway contract only match the standard contract.
    \item In the deterministic first-stage process, both contracts can outperform the standard contract and extract  welfare from the informative intermediate states.
\end{itemize}

These findings suggest that the pay-halfway contract primarily utilizes Type II information to infer the agent’s initial action and incentivize desirable behavior. In contrast, the terminate-halfway contract makes use of both Type I and Type II information: it not only deters poor initial actions but also allows the principal to exit early from low-quality intermediate states that are likely to result in unfavorable outcomes.

Finally, we compare the two newly proposed contracts directly. While the terminate-halfway contract can utilize richer information, our analysis shows that it does not always outperform the pay-halfway contract across all settings.


\subsection{Related Works}

Contract theory, as a fundamental discipline within economics, finds application across diverse contexts. It also boasts a rich historical lineage dating back to seminal works such as \cite{holmstrom1979moral,grossman1992analysis}. For an introductory exploration of contract theory, interested readers are directed to the comprehensive volumes  \cite{bolton2004contract} and the introductory chapters provided by \cite{caillaud2000hidden}, which offer foundational insights.

 \noindent\textbf{Contract Design Without States.} Recent research in contract theory spans a wide array of topics. Without introducing states, studies such as \cite{dutting2019simple}, \cite{guruganesh2021contracts}, \cite{guruganesh2023power}, and \cite{castiglioni2022designing} delve into the effectiveness of diverse contract types. For instance, \cite{dutting2019simple} conduct an analysis of linear contracts, demonstrating their optimality in addressing distributional intricacies. In contrast, \cite{guruganesh2021contracts} and \cite{guruganesh2023power} compare various contract types within scenarios characterized by unknown outcome distributions. Additionally, \cite{castiglioni2022designing} investigate the potency of randomization, focusing on the efficient design of random contract menus.
However, it is crucial to emphasize that within the framework of these previous works, there is a notable absence of intermediate states pertaining to the actions of agents. This distinction highlights a gap in the literature that our research aims to address by incorporating intermediate states into contract design. 

\noindent\textbf{Contract Theory in MDPs.}
The literature most closely related to our work studies principal--agent problems in Markov Decision Processes (MDPs). Early work by \cite{fudenberg1990short} and \cite{caillaud2000hidden} analyzes dynamic contracts with hidden actions in stochastic environments and characterizes optimal contracts via dynamic programming. \cite{zhang2008dynamic} consider a related MDP framework with different cost and reward structures and hidden states instead of hidden actions. More recent studies include \cite{zhang2021automated}, which examine repeated interactions where the principal conditions decisions on the agent’s strategic state reports, and \cite{zhang2022efficient,zhang2022planning}, which focus on MDPs in which the agent can terminate the process but does not control actions. \cite{bollini2024contracting} show that Markovian policies are generally suboptimal when agents take costly actions and propose efficient algorithms for computing optimal policies. Finally, \cite{wu2024contractual} study a contractual reinforcement learning setting in which a learning principal influences the agent’s policy through state-contingent payments.

Another research strand focuses on MDP-based contracts under global budget constraints rather than profit-contingent payments. Several studies \cite{zhang2008value,yu2022environment,ben2024principal,zhang2009policy} explore how reward shaping  can augment the principal's profit while considering observable actions. Specifically, \cite{zhang2008value} formulate a mixed-integer linear program, \cite{yu2022environment} consider a time-inconsistent biased agent, \cite{zhang2009policy} study the induction of a pre-specified target policy, and \cite{ben2024principal} focus on instances where polynomial time approximation
algorithms can be developed.

Compared with our work and most studies on MDP-based contract models, the core objectives differ fundamentally. MDP-based contract works generally aim to design an optimal contract within a fully dynamic MDP environment. In contrast, our work treats intermediate states as intermediate information and investigates whether and how access to this information benefits the principal relative to scenarios where such information is not leveraged.

The modeling assumptions also differ. We assume that the principal receives reward only upon reaching the final outcome, whereas MDP-based contract models typically accumulate rewards across states. Moreover, our model adopts a non-recurrent, two-step structure in which the intermediate state represents intermediate information. By contrast, MDP-based contract models may involve recurrent transitions and rely on structural assumptions suited to long-horizon, repeatedly interacting environments.

Intuitively, our model captures a single delegation process in which information is revealed along the way but rewards are realized only at the end. In contrast, most MDP-based contract models reflect successive delegation processes, where the principal commits to contracts and receives rewards at each step. Our framework therefore aligns more naturally with practical delegation settings that involve intermediate information but do not rely on the full machinery of dynamic MDP-based contract models.


\section{Setup}\label{sec_setup}
In this section, we introduce our new model, the two-stage delegation process. For those interested, the standard contract instance can be found in Appendix.

\subsection{Model}
A two-stage delegation process involves a principal and an agent. The principal assigns a task with $M$ possible outcomes to the agent, who must take actions to complete it. The process includes $S$ intermediate states rather than directly producing an outcome.
The agent first chooses one of $N_1$ initial actions, each incurring a cost $c_i$. Upon choosing action $i$, the agent transitions to an intermediate state $s$ according to a probability distribution $\mathcal{F}_i$, where $F_{i,s}$ denotes the probability of reaching intermediate state $s$.
At each intermediate state $s$, the agent selects one of $N_2$ final actions, each with an associated cost $c^s_j$. Selecting final action $j$ at the intermediate state $s$ leads to an outcome based on a probability distribution $\mathcal{F}^s_j$, with $F^s_{j,m}$ representing the probability of outcome $m$.
There exists a null initial action, and at every intermediate state, there is also a corresponding null final action, each incurring zero cost, i.e., $c_i = 0$ and $c^s_j = 0$ for the respective null actions.
The principal's reward depends solely on the realized outcome: a reward $r_m$ is received if the outcome $m$ occurs. The principal observes the intermediate state and the outcome, but not the agent's chosen initial nor the final action.
The expected reward of taking the final action $j$ in intermediate state $s$ is given by $R^s_j = \sum_{m=1}^M F^s_{j,m} r_m$.

An action profile $\boldsymbol{a} = (i, j_1, j_2, \dots, j_S)$ specifies that the agent first chooses the initial action $i$, and then selects the final action $j_s$ when transitioning to intermediate state $s$. The expected reward of profile $\boldsymbol{a}$ is calculated as $$R_{\boldsymbol{a}}=\sum_{s=1}^SF_{i,s}R^s_{j_s}.$$  The expected cost is calculated as $$c_{\boldsymbol{a}}=c_i+\sum_{s=1}^SF_{i,s}c^s_{j_s}.$$ The expected welfare of profile $\boldsymbol{a}$ is defined as $R_{\boldsymbol{a}} - c_{\boldsymbol{a}}$.
The maximal welfare $\mathcal{W}$ is $\W=\max_{\boldsymbol{a}}\{R_{\boldsymbol{a}}-c_{\boldsymbol{a}}\}$.

\subsection{Contracts and Agent's Behavior}
The principal's objective remains consistent with that of the standard contract instance, aiming to incentivize the agent to take actions advantageous to the principal. The principal achieves this by designing contracts that pledge transfers upon specific occurrences to incentivize agent behavior. 

In this section, we will initially outline various contracts. Subsequently, we will delineate the agent's behavior. Finally, we provide illustrative examples to aid understanding and demonstrate the practical application of these contract structures.

\subsubsection{Contracts}
We begin by introducing the standard contract and one of its special cases: the linear contract. Neither of these contracts utilizes information about intermediate states.

\textbf{Standard contract.} A standard contract $\boldsymbol{t}$ is an $M$-dimensional vector, where the $m$-th non-negative entry $t_m$ denotes the transfer made to the agent if outcome $m$ is realized. We denote the expected profit under the optimal standard contract by $\Single$.

\textbf{Linear Contract.}
A linear contract is a specific form of standard contract characterized by a fraction parameter $\alpha$. Under this contract, the agent receives a transfer of $\alpha r_m$ if outcome $m$ occurs. We denote the expected profit under the optimal linear contract by $\Linear$.

Next, we introduce two novel contract types designed to incorporate intermediate-state information: the pay-halfway contract and the terminate-halfway contract. These two contracts employ seemingly opposite mechanisms: one introduces additional transfers at intermediate states, while the other allows for early termination of the delegation process upon observing certain intermediate states.

\textbf{Pay-halfway Contract.}
The pay-halfway contract allows the principal to issue non-negative transfers based on intermediate states in addition to final outcomes. It is represented by a $(S+M)$-dimensional vector $(\boldsymbol{s}, \boldsymbol{t})$, where:
\begin{itemize}
\item $\boldsymbol{s} \in \mathbb{R}_{\geq 0}^S$ specifies the transfer $s_s$ upon reaching intermediate state $s$, and

\item $\boldsymbol{t} \in \mathbb{R}_{\geq 0}^M$ specifies the transfer $t_m$ upon outcome $m$.
\end{itemize}
We denote the expected profit under the optimal pay-halfway contract by $\M$.

\textbf{Terminate-halfway Contract.} The terminate-halfway contract allows the principal to end the delegation process upon observing certain intermediate states. It is defined by a pair $(\boldsymbol{t}, \mathcal{S})$, where:
\begin{itemize}
    \item  $\boldsymbol{t} \in \mathbb{R}_{\geq 0}^M$ denotes the transfer $t_m$ for each final outcome $m$, and
    \item $\mathcal{S} \subseteq \{1, 2, \dots, S\}$ specifies the set of intermediate states at which the process is terminated.
\end{itemize}
If an intermediate state $s \in \mathcal{S}$ is reached, the delegation process ends immediately, and both the principal and the agent receive zero payoff. The expected profit under the optimal terminate-halfway contract is denoted by $\T$.


\subsubsection{Agent's Behavior}We impose the standard assumptions of incentive compatibility (IC) and individual rationality (IR) from contract theory. In our setting, the IR constraint is readily satisfied due to the availability of null actions, which guarantee the agent a nonnegative outside option.

To analyze the IC constraint, note that once the principal commits to a contract, the agent faces a sequential decision problem. The agent selects actions to maximize her expected utility, and the optimal action profile can be characterized by backward induction. The agent first computes the optimal final action and corresponding utility for each intermediate state. Taking these into account, along with possible transfers at intermediate states and the risk of early termination, she then selects the optimal initial action. 

Specifically, given a contract:
(1) $\boldsymbol{t}$ (standard),
(2) $(\boldsymbol{s}, \boldsymbol{t})$ (pay-halfway), or
(3) $(\boldsymbol{t}, \mathcal{S})$ (terminate-halfway), the agent first determines the final optimal action $j^*_s$ for each intermediate state $s$ as:
$$j^*_s\in \argmax_{j} \sum_{m=1}^M F^s_{j,m}t_m-c^s_j.$$ 

Let $U^s$ denote the maximum utility the agent can obtain at intermediate state $s$: $$U^s=\max_{j}\sum_{m=1}^M F^s_{j,m}t_m-c^s_j.$$ Then, the agent selects the initial action $i^*$ by considering the utility $U^s$, the intermediate transfers $s_s$, and the risk of termination:  $$i^*\in \argmax_i\sum_{s=1}^S \mathbbm{1}(s\notin \mathcal{S})F_{i,s}(U^s+s_s)-c_i.$$ 
Upon reaching an intermediate state $s \notin \mathcal{S}$, the agent selects final action $j^*_s$. We assume that in case of ties, the agent breaks ties in favor of the principal. \footnote{This can be achieved by utilizing arm recommendations.} Accordingly, the expected payment from the principal to the agent is: $$T=\sum_{s=1}^S \mathbbm{1}(s\notin \mathcal{S})F_{i^*,s}(F_{j_s^*,m}^st_m+s_s).$$

\subsubsection{Examples of Newly Proposed Contracts}\label{sec_example}
In this section, we provide examples of the newly proposed contracts to demonstrate their practical use and deepen the understanding of our model.
\begin{table}
\centering
\begin{tabular}{cccc}
  \toprule
         & \textbf{State} $1$, $\textit{fail}$ & \textbf{State} $2$, $\textit{pass}$                                               \\ \midrule
  \textbf{Initial Action} $1$, $\textit{effort}$ & \multirow{2}{*}{$0.1$} & \multirow{2}{*}{$0.9$}\\ Cost $1.8$ \\
  \textbf{Initial Action} $2$, $\textit{null}$ & \multirow{2}{*}{$1$} &  \multirow{2}{*}{$0$} \\ Cost $0$\\ \midrule
  state $1$, $\textit{fail}$        & \textbf{Outcome} $1$, $\textit{fail}$ & \textbf{Outcome} $2$, $\textit{success}$   \\
   & Reward $0$           &  Reward $5$                                            \\ \midrule
  \textbf{Final Action} $1$, $\textit{effort}$ & \multirow{2}{*}{$0.2$} & \multirow{2}{*}{$0.8$}\\ Cost $2$ \\
  \textbf{Final Action} $2$, $\textit{null}$ & \multirow{2}{*}{$0.9$} &  \multirow{2}{*}{$0.1$} \\ Cost $0$\\ \midrule
state  $2$, $\textit{pass}$ &  &   \\ \midrule
  \textbf{Final Action} $1$, $\textit{effort}$ & \multirow{2}{*}{$0$} & \multirow{2}{*}{$1$}\\ Cost $1$ \\
  \textbf{Final Action} $2$, $\textit{null}$ & \multirow{2}{*}{$0$} &  \multirow{2}{*}{$1$} \\ Cost $0$\\ 
  \bottomrule
\end{tabular}
\caption{The instance of Example \ref{example_median}
}\vspace{-0.35 in}
\label{table_median}
\end{table}
\begin{example}[Pay-halfway Contract] \label{example_median}
If a teacher assigns a programming task to a student, her primary concern is the final outcome: success or failure. The teacher receives a reward of $r = 5$ only if the student successfully completes the program (denoted as "success"). There is a mid-term presentation to evaluate the student's progress, with outcomes either "pass" or "fail". If the student delivers an outstanding presentation, reaching the intermediate state "pass", the likelihood of achieving "success" increases. Otherwise, the program is more likely to fail.

Initially, the student has two choices: to exert effort or not, denoted as "effort" and "null". Exerting effort incurs a cost of $1.8$ and gives a high probability ($p = 0.9$) of reaching the "pass" state. Choosing "null" inevitably leads to failing the presentation. After the mid-term presentation, upon reaching an intermediate state, the student again faces two choices: to exert effort or not.
If the student passes the presentation, either action leads to a successful final outcome, but "effort" incurs an additional cost of $1$. If the student fails the presentation, the "effort" action incurs a cost of $2$ and has a $0.8$ probability of leading to success, while the "null" action results in only a $0.1$ probability of success.

In this scenario, although the teacher doesn't directly benefit from the student's good performance in the mid-term presentation, she may still choose to reward the student for doing well in order to incentivize the student’s initial effort. For instance, the teacher might promise a reward of $2$ if the student achieves "${\textit{pass}}$" in the mid-term presentation and $0.1$ if the student successfully completes the program. This payment arrangement resembles a pay-halfway contract $(\boldsymbol{s},\boldsymbol{t})$, where $\boldsymbol{s}=(0,2), \boldsymbol{t}=(0,0.1)$.
Given this "contract", the student employs backward induction. At the "$\textit{pass}$" state, she chooses the action "$\textit{null}$", resulting in an expected utility of $U_1=0.1$. At the "$\textit{fail}$" state, she also chooses the null action, yielding an expected utility of $U_2=0.01$. When determining the initial action, if she exerts effort, the expected utility is $(U_1+2)p+U_2(1-p)-c=0.091$. Otherwise, the expected utility is $0.01$. Consequently, the student opts to exert effort initially and takes the null action regardless of the state. The expected profit for the principal is $p(r-0.1-2)+(1-p)(0.1(r-0.1))=2.659$.
\end{example}

\begin{table}
\centering
\begin{tabular}{cccc}
  \toprule
         & \textbf{State} $1$, $\textit{bad}$ & \textbf{State} $2$, $\textit{well}$                                               \\ \midrule
  \textbf{Initial Action} $1$, $\textit{effort}$ & \multirow{2}{*}{$0.01$} & \multirow{2}{*}{$0.99$}\\ Cost $8$ \\
  \textbf{Initial Action} $2$, $\textit{null}$ & \multirow{2}{*}{$1$} &  \multirow{2}{*}{$0$} \\ Cost $0$\\ \midrule
  state $1$, $\textit{bad}$        & \textbf{Outcome} $1$, $\textit{failure}$ & \textbf{Outcome} $2$, $\textit{success}$   \\
   & Reward $0$           &  Reward $10$                                            \\ \midrule
  \textbf{Final Action} $1$, $\textit{effort}$ & \multirow{2}{*}{$0$} & \multirow{2}{*}{$1$}\\ Cost $1$ \\
  \textbf{Final Action} $2$, $\textit{null}$ & \multirow{2}{*}{$0$} &  \multirow{2}{*}{$1$} \\ Cost $0$\\ \midrule
state  $2$, $\textit{well}$ &  &   \\ \midrule
  \textbf{Final Action} $1$, $\textit{effort}$ & \multirow{2}{*}{$0.4$} & \multirow{2}{*}{$0.6$}\\ Cost $4$ \\
  \textbf{Final Action} $2$, $\textit{null}$ & \multirow{2}{*}{$0.9$} &  \multirow{2}{*}{$0.1$} \\ Cost $0$\\ 
  \bottomrule
\end{tabular}
\caption{The instance of Example \ref{exa_terminal}
}\vspace{-0.35 in}
\label{table_terminal}
\end{table}
\begin{example}[Terminate-halfway Contract]\label{exa_terminal}
When an employer delegates a task to a worker, the employer's primary concern is the completion of the task. However, there may be a need to estimate the progress of the task. If the worker performs poorly during an interim evaluation, the employer might terminate the delegation process and reassign the task to another worker.

Specifically, there are two potential outcomes: "success" and "failure", along with two intermediate states: "well" and "bad". The employer is rewarded with $r = 10$ only upon the successful completion of the task. The worker has two initial actions and two final actions: "effort," which entails exertion at a cost of $8$ and leads to the "well" state with a high probability of $p = 0.99$, and "null," representing initial action with zero cost that leads to the "bad" state. At the "well" state, both final actions result in success. In contrast, at the "bad" state, choosing "effort" results in success with a probability of $0.6$ at an additional cost of $4$, while selecting "null" results in failure with a probability of $0.9$.

In this scenario, the employer may promise to give the worker a payment of $8.2$ if the task is completed successfully, i.e., if the "success" outcome is reached. Additionally, the employer will terminate the delegation process if the worker performs poorly at the interim evaluation, i.e., if the "bad" intermediate state is reached. This can be viewed as a terminate-halfway contract represented by $(\boldsymbol{t},\mathcal{S})$ where $\boldsymbol{t}=(0,8.2)$ and $\mathcal{S}=\{bad\}$. Regarding this contract, the worker first computes that at the "well" state, she will take the final action "null," resulting in an expected utility of $8.2$. Then, she recognizes that she will get zero utility regardless of the final action she takes at the "bad" state. Consequently, by simple computation, she decides to take the initial action "effort." In this way, the employer earns the expected profit of $1.8$.
\end{example}

\subsection{Specialized Two-stage Delegation Process}\label{sec_special}

Intermediate states primarily convey two types of information: 

\begin{itemize}
    \item[(I)] the intrinsic quality of the state itself, specifically, the likelihood that subsequent final actions taken from this intermediate state will lead to valuable outcomes; and 
    \item[(II)] the agent’s initial action choice. 
\end{itemize}

In this section, we introduce several specialized two-stage delegation processes that differ in the type and extent of information embedded in intermediate states. Our objective is to examine whether, and which type of, information benefits the principal. More specifically, we aim to assess whether and how the newly proposed contracts can effectively leverage these different types of information. 

 In the tree two-stage process, each outcome can be reached through only one specific intermediate state. Consequently, the intermediate state appears to offer no additional information beyond what is ultimately revealed by the final outcome.  
 
 In the random first-stage  process, there is only one initial action, which may lead to multiple intermediate states. These intermediate states do not reveal any additional information about the agent’s initial choice (since there is only one). However, the distribution of final actions leading to outcomes can vary across different intermediate states. Some intermediate states may be more likely to result in favorable outcomes than others.  Therefore, intermediate states in this setting primarily convey information about the intrinsic quality of the states themselves (Type I information). 
 
 Finally, in the deterministic first-stage process, each intermediate state corresponds uniquely to a specific initial action. As a result, observing the intermediate state fully reveals the agent’s initial decision and may also convey information about the quality of ensuing outcomes. Hence, intermediate states in this case encapsulate both types of information and are potentially highly informative.

We now proceed to formally define these specialized two-stage delegation processes:

\begin{itemize}
    \item \textbf{Tree Two-stage Process}: In this process, each final outcome can only be realized through a unique intermediate state. Formally, for each outcome $m$, there exists at most one intermediate state $s$ such that $\sum_{j\in[N_2]} F^s_{j,m} > 0$. 
    \item  \textbf{Stochastic First-stage Process}: In this process, there is only one initial action, i.e., $N_1 = 1$.
    \item \textbf{Deterministic First-stage  Process}: In this process, each initial action deterministically leads to exactly one intermediate state. Specifically, for every initial action $i$, there exists an  intermediate state $s$ such that $F_{i,s} = 1$.
\end{itemize}

    \section{Upper Bound Results}\label{general}
In this section, we briefly explore the role and impact of information embedded in intermediate states within our framework, providing several upper-bound results.

We first analyze the extent to which access to intermediate-state information improves the optimal welfare compared to the profit achievable under a standard contract. By establishing upper bounds through comparisons between the welfare and the profit under the optimal standard contract, we derive limits on the maximum profit ratio that our newly proposed contracts can achieve relative to the standard contract.

We show that, in the general two-stage delegation process, the total welfare is at most $O(SN_1N_2)$ times the profit. We also consider the specialized two-stage delegation process and derive corresponding upper bound results.
\begin{theorem}\label{theorem_median_up}
Consider a two-stage delegation process with $S$ intermediate states, $N_1$ initial actions, and $N_2$ final actions, the welfare is no more than $O(SN_1N_2)$ times the profit obtained by the optimal standard contract, i.e., 
\begin{equation}
    \W\leq O(SN_1N_2) \cdot \Single.
\end{equation}
 \end{theorem}

\begin{proof}[Proof Sketch]

   Firstly, given that $ \Single \geq \Linear$, we can constrain the ratio of welfare and profit obtained through the optimal standard contract by the ratio of welfare and profit obtained through the optimal linear contract, i.e. $\frac{\W}{\Single} \leq \frac{\W}{\Linear}$. Thus, our focus narrows down to analyzing the ratio of welfare and profit achieved through the optimal linear contract, denoted as $r = \frac{\W}{\Linear}$.

Consequently, we examine the performance of the optimal linear contract, starting with the aspect of breakpoints. A breakpoint $\alpha$ marks the threshold at which the agent's chosen action undergoes a transition. It can be obtained that the welfare can be bounded in terms of the number of breakpoints and the profit under the optimal standard contract. Specifically, the following inequality holds:
\[
\W \leq n \Single
\]
where \(n\) is the number of breakpoints.

There are at most $N_1(N_2)^S$
  possible action profiles, so at most that many breakpoints, but this is a loose bound. 
  For each intermediate state $s$, there can be at most $N_2$ breakpoints (one for each possible final action), denoted as $\alpha_1^s, \alpha_2^s, \ldots, \alpha_{N_2}^s$. Sort these breakpoints in increasing order and relabel them accordingly, resulting in $\alpha_1, \alpha_2, \ldots, \alpha_m$. Consequently, there are at most $SN_2$ breakpoints in total for the final actions..

 Between any two adjacent breakpoints $\alpha_i$ and $\alpha_{i+1}$, the choice of the agent's final action remains unchanged, and the expected reward of every initial action remains constant as well. Thus, we can interpret this process as a new standard contract instance, with each intermediate state acting as an outcome. Under this perspective, there are at most $N_1$ breakpoints between $\alpha_i$ and $\alpha_{i+1}$. Consequently, there are at most $SN_1N_2$ breakpoints in total.

In conclusion, with $SN_1N_2$ breakpoints, the ratio of welfare and profit achieved through the optimal linear contract, denoted as $r$, is bounded by no more than $SN_1N_2$. Thus, we arrive at the desired result.
\end{proof}
\vspace{-0.1 in}
\begin{claims}\label{claim_special}
     Consider a two-stage delegation process with $S$ intermediate states, $N_1$ initial actions, and $N_2$ final actions: 

  \begin{itemize}
\item[(1)] In a deterministic first-stage process, welfare is bounded by $\W \leq O(\min\{S, N_1\} \cdot N_2) \cdot \Single$.
\item[(2)] In a stochastic first-stage process, welfare is bounded by $\W \leq O(S N_2) \cdot \Single$.
\end{itemize}

\end{claims}

Notably, total welfare provides only a loose upper bound on the principal’s attainable profit under any contract. Since welfare is defined as the reward minus the agent’s cost, full extraction would leave the agent with negligible utility, which may render such contracts infeasible (e.g., involving negative transfers). 

The following claim demonstrates that, in the tree two-stage process, neither of the newly proposed contracts outperforms the standard contract. This result is consistent with the observation that, in the tree two-stage process, the intermediate state provides no additional information beyond what is ultimately revealed by the final outcome.

\begin{claims}\label{theorem_tree}
    Consider a tree two-stage process with $S$ intermediate states, $N_1$ initial actions, and $N_2$ final actions, the optimal  pay-halfway contract and the optimal terminate-halfway can only yield the same profit with the optimal standard contract, i.e.,
    \begin{equation}
    \M = \T=\Single.
\end{equation}    
\end{claims}
\section{Pay-halfway Contract}\label{sec_median}
In this section, we conduct a detailed analysis of the pay-halfway contract, with a focus on the central questions: \emph{Does the pay-halfway contract leverage information embedded in intermediate states to improve the principal’s profit, and if so, how and to what extent? }Specifically, we compare the performance of the pay-halfway contract with that of the standard outcome-based contract. Our aim is to determine whether the pay-halfway contract can generate strictly higher profit than the standard contract, and to identify the conditions under which such improvement arises, as well as the magnitude of its potential advantage.

We begin by referencing Example \ref{example_median} to affirmatively answer our first question. That is, the pay-halfway contract can indeed generate higher profits compared to a standard contract. In Example \ref{example_median}, while the optimal standard contract $\boldsymbol{t}=(0,\frac{20}{9})$ yields an approximate profit of $2.53$, the pay-halfway contract we provide in the Example ($(\boldsymbol{s},\boldsymbol{t})$, where $\boldsymbol{s}=(0,2), \boldsymbol{t}=(0,0.1)$) achieves an expected profit of $2.659$. This highlights the effectiveness of the pay-halfway contract.

The remainder of the section is structured as follows. We first present Theorem \ref{theorem_median_cost}, which shows that the pay-halfway contract can substantially reduce the principal’s total payment, thereby illustrating its effectiveness relative to the standard contract. We then analyze a series of specialized delegation processes to investigate the conditions under which the pay-halfway contract outperforms the standard contract, and to what extent.

The following theorem presents an instance that highlights the power of the pay-halfway contract in reducing payments: the ratio of payments required to incentivize the optimal action profile under the standard contract versus the pay-halfway contract can become arbitrarily large.

\begin{theorem}\label{theorem_median_cost}
    There is a two-stage delegation process with $S$ intermediate states, $N_1$ initial actions, and $N_2$ final actions, where the payment ratio between incentivizing  the same action profile $\boldsymbol{a}^*$ by the standard contract and the pay-halfway contract approaches infinity. The action profile $\boldsymbol{a}^*$ represents the optimal action profile in terms of maximizing welfare.
\end{theorem}
\begin{proof}[Proof Sketch]
We construct an example with three intermediate states and two outcomes.  In one state, the final action yields a high reward at zero cost; in another, the high reward is achievable but requires a transfer to cover a cost $c$; in the third state, all final actions lead to zero reward. There are two initial actions: one initial action assigns a higher probability $p$ to the high-reward zero-cost state and the remaining probability to the zero-reward state, while the other assigns a lower probability $q<p$ to the high-reward zero-cost state and allocates the rest to the costly state. The welfare-optimal profile chooses the second initial action that involves visiting the costly state.

Under a standard contract, incentives are based only on final outcomes. To incentivize the optimal action profile, the payment for the high-reward outcome must cover the agent’s cost and compensate for incentive constraints arising from the difference in transition probabilities. 

In contrast, a pay-halfway contract  enables the principal to directly cover the agent’s cost at the critical intermediate state with a much smaller total payment, while still preserving incentive compatibility.

By carefully choosing the transition probabilities, the payment required under the standard contract grows unbounded, while the payment under the pay-halfway contract remains bounded. This results in the ratio of payments approaching infinity, demonstrating the claimed separation.
\end{proof}

\subsection{Results in Specialized Delegation Process}
We now analyze how the pay-halfway contract leverages the information embedded in intermediate states by comparing its performance to the standard contract across different specialized delegation processes.

We first show that in the stochastic first-stage process, the optimal pay-halfway contract yields the same profit as the optimal standard contract. 
\begin{claims}\label{theorem_median_r}
     Consider a stochastic first-stage  process  with $S$ intermediate states, $N_1$ initial actions, and $N_2$ final actions, the optimal  pay-halfway contract yields the same profit with the optimal standard contract, i.e.,
    \begin{equation}
    \M = \Single.
    \end{equation}
\end{claims}

In contrast, in the deterministic first-stage process, the pay-halfway contract can extract nearly the entire welfare and significantly outperform the standard outcome-based contract. 

In a deterministic first-stage  process, each initial action leads to only one intermediate state. When employing a pay-halfway contract, the agent can simply be incentivized to undertake costly initial actions by offering transfers to the corresponding intermediate state. However, with a standard contract, if we aim to incentivize costly initial actions, we can only raise payments on certain outcomes. Yet, this might inadvertently increase the expected utility for other states, consequently affecting other initial actions as well. 
Hence, we construct an instance where the pay-halfway contract can specifically incentivize optimal actions, while the standard contract must increase rewards for all initial actions, resulting in inefficient use of payments.

\begin{theorem}\label{theorem_median_exa}
    There is a deterministic first-stage process  with $N_1+1$ intermediate states, $N_1+1$ initial actions, and $N_2+1$ final actions, where the profit obtained by the optimal pay-halfway contract is more than $\Omega(N_1N_2)$ times the profit obtained by the optimal standard contract, i.e., 
\begin{equation}
    \M\ge \Omega(N_1N_2) \cdot \Single.
\end{equation}
\end{theorem}
\begin{proof}[Proof Sketch]

Consider two possible outcomes: a zero-reward outcome \( r_1 = 0 \), and a high-reward outcome \( r_2 = r \), where \( r \) is chosen to be sufficiently large. There are \( N_1 + 1 \) initial actions, each deterministically leading to a distinct intermediate state. Among these, a special initial action \( a_{N_1+1} \) incurs a large cost
$
c = \sum_{k=1}^{(N_1+1)N_2 + 1} \lambda^k
$
and leads to the state \( s_{N_1+1} \). All other initial actions incur zero cost and lead to intermediate states \( s_1, \ldots, s_{N_1} \), respectively.

Each intermediate state has \( N_2 + 1 \) final actions. In particular, in state \( s_{N_1+1} \), all final actions incur no additional cost and yield an expected reward of
$
R = c + (N_1 + 1)N_2 + 1.
$
For each \( s \in [N_1] \) and final action \( a_j^s \) with \( j \in [N_2] \), the associated cost is
$
c_j^s = \sum_{k=1}^{sN_2 + j} \lambda^k,
$
and the probability of achieving the high-reward outcome is set such that the expected reward is
$
R_j^s = c_j^s + sN_2 + j.
$

A pay-halfway contract that pays \( c \) at state \( s_{N_1+1} \) can fully extract the total welfare, which is
$
\mathcal{M} = \mathcal{W} = R - c = (N_1 + 1)N_2 + 1.
$
In contrast, an optimal standard-contract mechanism can achieve at most the same profit as linear contracts. Due to the exponential nature of the cost structure, the maximum achievable profit under such contracts is bounded by \( O(1) \).

Therefore, the profit ratio satisfies: 
$
\frac{\mathcal{M}}{\Single} \ge \frac{(N_1 + 1)N_2 + 1}{O(1)} = \Omega(N_1 N_2).
$
\end{proof}

From Claim \ref{claim_special}, we observe that the profit under the optimal pay-halfway contract does not exceed $O(N_1N_2)$ times the profit achievable under the optimal standard contract, as the profit is bounded above by the welfare. Therefore, the performance of the pay-halfway contract actually matches the tight bound.

\begin{remark}
    In conclusion, we find that the pay-halfway contract primarily leverages intermediate-state information to infer the agent’s initial action choice and thereby incentivize desirable initial behavior. Consequently, this contract structure is particularly effective in deterministic first-stage process, where intermediate states directly reflect the agent’s initial actions, but is less suitable for stochastic first-stage process, where such information is absent. 
\end{remark}

\section{Terminate-halfway Contract}\label{sec_terminal}
In this section, we analyze the terminate-halfway contract. Similarly, we investigate whether and how this contract leverages information embedded in intermediate states. 

In Example \ref{exa_terminal}, the optimal standard contract $\boldsymbol{t}=(0,8)$ yields the expected profit of $1.2$ which is strictly less than the expected profit of the terminate-halfway contract we mentioned in the example. Thus, the terminate-halfway contract provides a higher expected profit compared to the optimal standard contract, demonstrating the effectiveness of incorporating intermediate state information into the contract design.

We further analyze the aforementioned question by comparing its performance to that of the standard outcome-based contract across various specialized delegation processes.

\subsection{Results in Specialized Delegation Process}
We present several tight counterexamples within specialized delegation processes to illustrate the effectiveness of the terminate-halfway contract. The results simultaneously demonstrate which types of information the terminate-halfway contract can utilize and the extent to which these information types can enhance its performance. We focus on high-level insights, with full technical details provided in the Appendix.

We start with a deterministic first-stage process to show the significant advantage of terminate-halfway contracts.
\begin{theorem}\label{theorem_terminal_d_exa}
    There is a deterministic first-stage process  with $N_1+1$ intermediate states, $N_1+1$ initial actions, and $N_2+1$ final actions, where the profit obtained by the terminate-halfway contract is more than $\Omega(N_1N_2)$ times the profit obtained by the optimal standard contract, i.e., 
\begin{equation}
    \T\ge \Omega(N_1N_2) \cdot \Single.
\end{equation}
\end{theorem}
\begin{proof}[High-Level Idea]  The terminate-halfway contract permits the delegation process to be terminated if certain unfavorable intermediate states are observed, resulting in an expected utility of zero for the agent. Consequently, this contract can potentially incentivize the agent to select specific initial actions, particularly in deterministic settings.  In the deterministic setting, the terminate-halfway contract can motivate the agent to choose a specific initial action by blocking all intermediate states except the one that results from the desired initial action. In contrast, this level of control is not achievable under a standard contract. Therefore, we construct an instance, similar to that in Theorem \ref{theorem_median_exa}, in which any standard contract can obtain at most a $\frac{1}{\Omega(N_1N_2)}$ fraction of the total welfare, whereas the terminate-halfway contract is able to extract the full welfare.
\end{proof}

Next, we consider a stochastic first-stage process to further illustrate the superior performance of terminate-halfway contracts.
\begin{theorem}\label{theorem_terminal_r_example}
    There is a stochastic first-stage  process  with $2S+1$ intermediate states, and $N_2+1$ final actions, where the profit obtained by the terminate-halfway contract is more than $O(S N_2)$ times the profit obtained by the optimal standard contract, i.e., 
\begin{equation}
\T\ge \Omega(S 
    N_2) \cdot \Single.
\end{equation}
\end{theorem}
\begin{proof}[High-Level Idea] 
The effectiveness of a terminate-halfway contract stems from the principal's ability to end the delegation process when adverse intermediate states are observed. This allows the principal to "block" states that would result in negative profits, thereby maximizing overall profit.
Consider a scenario with one positive reward outcome and zero-reward outcomes achievable only through optimal final actions in certain states. Incentivizing these actions via the positive outcome can lead to wasteful payments. Furthermore, some states lead exclusively to zero-reward outcomes regardless of the final action. With a terminate-halfway contract, the principal can incentivize the optimal final actions through transfers conditioned on zero-reward outcomes, while blocking states that might result in negative profits. In contrast, a standard contract would either require excessive payments to incentivize desirable actions through the transfer on the high-reward outcome or result in negative profits from undesirable states.
Therefore, the terminate-halfway contract enables the principal to extract nearly the full welfare by strategically managing incentives and "blocking" negative outcomes. Conversely, a standard contract forces the principal to choose between incurring wasteful payments or facing negative profits from adverse states.
\end{proof}

\begin{remark}
    In summary, the strength of the terminate-halfway contract lies in its ability to reduce excessive payments by allowing the principal to terminate the delegation process upon observing unfavorable signals (i.e., undesirable intermediate states). It can also incentivize the agent to avoid poor initial actions.  This contract is particularly effective in the deterministic first-stage process, where intermediate states reveal the agent’s initial action, thereby enabling the principal to encourage desirable initial action. Additionally, it performs well in stochastic first-stage process, where intermediate states may indicate the potential quality of outcomes resulting from final actions, allowing the principal to preemptively terminate low-value trajectories. Therefore, the terminate-halfway contract effectively leverages both types of intermediate-state information: the intrinsic quality of the state and the agent’s initial action.
\end{remark}

\section{Relationship Between Pay-halfway and Terminate-halfway Contract}\label{sec_separation}
In this section, we examine the relationship between the two newly proposed contract types.

Building on previous findings, we observe that the pay-halfway contract primarily leverages intermediate-state information to infer the agent’s initial action and to incentivize desirable initial behavior. Consequently, it tends to perform relatively poorly in the stochastic first-stage process, where intermediate states are less informative about the agent's initial choice. In contrast, the terminate-halfway contract appears to make more comprehensive use of intermediate-state information: it not only discourages undesirable initial actions but also allows the principal to avoid low-quality intermediate states that are likely to lead to poor outcomes. These observations naturally lead to the question: Does the terminate-halfway contract exhibit strictly greater effectiveness than the pay-halfway contract? More specifically, does it consistently outperform the pay-halfway contract across all settings?

\begin{theorem}\label{obs_seperation}
 There exists an instance where the pay-halfway contract outperforms the terminate-halfway contract.
\end{theorem}

\begin{proof}[Proof Sketch]
    We follow the same instance construction as in Theorem~\ref{theorem_median_cost}. The key idea is to compare the maximal achievable profit under the pay-halfway contract and the terminate-halfway contract for the optimal action profile.

For the terminate-halfway contract, there are two cases to consider: either some state is blocked, or no state is blocked. If no state is blocked, the contract effectively reduces to a standard contract, which is strictly suboptimal compared to the pay-halfway contract. Conversely, if a state is blocked, the payment necessary to incentivize the optimal action profile yields a strictly lower maximal profit than that of the pay-halfway contract, as the profit in the blocked state remains positive.

By comparing the profit expressions for both contracts, we show that the pay-halfway contract yields strictly higher profit under appropriate parameter settings. This demonstrates the superiority of the pay-halfway contract over the terminate-halfway contract in this setting.
\end{proof}
\vspace{-0.1 in}
This Theorem answers our question by showing that the pay-halfway contract can outperform the terminate-halfway contract in certain cases.  One possible explanation is that, although the terminate-halfway contract leverages richer information—both the agent’s initial action and the intrinsic quality of the intermediate state—while the pay-halfway contract mainly relies on the former, the pay-halfway contract may be more effective at incentivizing desirable initial actions.

\section{Conclusion}\label{sec_discussion}
This paper investigates the role and the impact of information during the delegation process by introducing a novel two-stage contract model with multiple intermediate states. 
To leverage the information revealed during the process, we propose two conceptually distinct contract types: the pay-halfway contract, which allows transfers based on both final outcomes and selected intermediate states; and the terminate-halfway contract, which permits early termination of the delegation process at certain intermediate states.
We conduct a comprehensive analysis of these contracts by comparing their performance to that of the optimal standard contract. This comparison highlights the relative efficiency and effectiveness of the newly proposed contracts that incorporate intermediate information, revealing their respective strengths and limitations. It thereby offers deeper insights into how each contract type exploits information and the practical implications for contract design.

Through this investigation, our work advances the theory of contract design by illustrating how intermediate information can be leveraged to improve performance in delegation settings.

\bibliographystyle{ACM-Reference-Format} 
\bibliography{ref}

\onecolumn
\newpage
\appendix
\section{Notation Table}
In this section, we summarize the notation introduced in Section~\ref{sec_setup} to facilitate the reader’s understanding of the model and subsequent analysis.
\begin{table}[ht]
\centering
\caption{Summary of Notation}
\label{tab:notation}
\begin{tabular}{ll}
\toprule
\textbf{Symbol} & \textbf{Meaning} \\
\midrule
$M$ & Number of possible final outcomes \\
$S$ & Number of intermediate states \\
$N_1$ & Number of initial actions available to the agent \\
$N_2$ & Number of final actions available at each intermediate state \\[0.5ex]

$c_i$ & Cost of initial action $i$ \\
$c^s_j$ & Cost of final action $j$ at intermediate state $s$ \\
$\mathcal{F}_i$ & Distribution over intermediate states induced by initial action $i$ \\
$F_{i,s}$ & Probability of reaching intermediate state $s$ after action $i$ \\
$\mathcal{F}^s_j$ & Outcome distribution induced by final action $j$ at state $s$ \\
$F^s_{j,m}$ & Probability of outcome $m$ under action $j$ at state $s$ \\[0.5ex]

$r_m$ & Principal's reward if outcome $m$ occurs \\
$R^s_j$ & Expected reward of final action $j$ at state $s$ \\
$\boldsymbol{a}=(i,j_1,\dots,j_S)$ & Action profile of the agent \\
$R_{\boldsymbol{a}}$ & Expected reward under action profile $\boldsymbol{a}$ \\
$c_{\boldsymbol{a}}$ & Expected cost under action profile $\boldsymbol{a}$ \\
$\mathcal{W}$ & Maximum expected welfare \\[0.5ex]

$\boldsymbol{t}=(t_1,\dots,t_M)$ & Transfers contingent on final outcomes \\
$\alpha$ & Fraction parameter in linear contracts \\
$\boldsymbol{s}=(s_1,\dots,s_S)$ & Transfers contingent on intermediate states \\
$\mathcal{S}$ & Set of intermediate states triggering termination \\[0.5ex]

$\Single$ & Expected profit under optimal standard contract \\
$\Linear$ & Expected profit under optimal linear contract \\
$\M$ & Expected profit under optimal pay-halfway contract \\
$\T$ & Expected profit under optimal terminate-halfway contract \\
\bottomrule
\end{tabular}
\end{table}

\section{Standard Contract Instance} \label{sec_standard}
In this section, we provide a brief overview of the classical model of contract design, referred to as the standard contract instance.

\textbf{Model.} A contract instance constitutes a principal and an agent. The agent has $N$ actions that can result in $M$ outcomes.  Action $i$ incurs a cost denoted as $c_i$ and induces a distribution $\mathcal{F}_i$ over outcomes, where $F_{i,m}$ represents the probability of outcome $m$ resulting from action $i$.  It is commonly assumed that there exists a null action with zero cost, i.e., $c_i=0$ for the corresponding null action.
The principal can get a reward $r_m$ if the outcome $m$ is realized. 
The principal can solely observe the realized outcome, but not the action undertaken.  Let $R_i$ denote the expected reward of action $i$, such that $R_i=\mathbb{E}_{m\sim\mathcal{F}_i}[r_m]$.  The expected welfare of action $i$ is the difference between the expected reward and the cost of taking action $i$, i.e., $R_i-c_i$. 

\textbf{Contract.} The principal's objective is to incentivize the agent to take actions that are advantageous to the principal. To achieve this, the principal devises a contract for the agent, wherein the principal commits to providing the agent with certain payments contingent upon the realized outcomes. Specifically, a contract $\boldsymbol{t}$ is represented as an $M$-dimensional vector, where the non-negative entry $t_m$ for outcome $m$ signifies the transfer to be granted to the agent upon the realization of that outcome. 

\textbf{Agent's Behavior.} Given a fixed contract $\boldsymbol{t}$, the agent will take the action that satisfies two conditions: (i) incentive compatibility (IC), meaning it maximizes the agent's expected utility among all actions, and (ii) individual rationality (IR), meaning it ensures the agent's expected utility is non-negative. Given that the contract assigns non-negative transfers and there includes a null-action, individual rationality is trivially satisfied. Consequently, the agent will choose an action $i^*(\boldsymbol{t})$ such that:
\begin{align*}
    i^*(\boldsymbol{t})\in \argmax_i \sum_{m=1}^M F_{i,m}t_m-c_i.
\end{align*}
Moreover, we adopt the common assumption, when there are ties, the agent will take actions in favor of the principal.  If agent will take action $i$ given contract $\boldsymbol{t}$, i.e., $i=i^*(\boldsymbol{t})$, we say the contract $\boldsymbol{t}$ incentivizes action $i$.

\textbf{Optimal Contract.} The principal's expected profit is defined as difference of the expected reward and the expected payment, i.e., 
\begin{align*}
    P(\boldsymbol{t})=\sum_{m=1}^MF_{i^*(\boldsymbol{t}),m}(r_m-t_m).
\end{align*}
The optimal contract can be determined using linear programming. Specifically, for each action $i$, we solve a linear program to find the contract that incentivizes action $i$ while minimizing the principal’s expected payment. Among all such contracts, the optimal contract is the one that yields the highest expected profit to the principal.

For each action $i$, the linear program (LP) aimed at incentivizing it features $M$ non-negative transfer variables $\{t_m\}$ and $N-1$ constraints to ensure incentive compatibility, as follows:
\begin{align*}
    \min \sum_{m=1}^M &F_{i,m}t_m,
    \\
    s.t. \sum_{m=1}^M F_{i,m}t_m-c_i&\ge \sum_{m=1}^M F_{i',m}t_m-c_{i'}, \quad \forall i'\ne i, i'\in [N],
    \\
    t_m&\ge 0 \quad \forall m\in[M].
\end{align*}

\section{Complexity Results}\label{app:complexity}
\begin{theorem}
For a two-stage delegation process with $S$ states, $N_1$ initial actions, and $N_2$ final actions, computing the optimal standard contract is APX-hard.
\end{theorem}
\begin{proof}
It is evident that the stochastic first-stage  process represents merely a specific case within our broader framework. The Bayesian agent-type model can, in fact, be viewed as an instance of this stochastic first-stage  process. Building on the findings from \cite{guruganesh2021contracts,castiglioni2021bayesian}, we can straightforwardly derive our conclusion from this observation.

\end{proof}

\begin{theorem}
    There exists an algorithm with time complexity of $ploy(S,N_1,N_2^S,M)$ to determine the optimal standard contract.
\end{theorem}
\begin{proof}
To determine the optimal standard contract, we propose a linear programming approach. The linear program consists of  $M$ variables that need to be solved.  These variables represent the transfer with corresponding  outcomes.  Additionally, there are at most $N_1N_2^S$  action profiles to incentivize.  Each action profile represents a specific combination of actions that the agents can take.

For any given action profile, we can compute the optimal contract with minimal payment to incentivize it using a linear program.

The linear program involves $M$ variables and $N_1-1+(N_2-1)S+M$ constraints. The constraints ensure that the contract is feasible and meets the incentive requirements for the agents.  Consequently, it can be solved in polynomial time, specifically $ploy(S,N_1,N_2,M)$. Since there are $N_1N_2^S$ possible action profiles, and we need to compute the optimal contract for each profile, the overall algorithm operates in polynomial time relative to the number of action profiles and the number of variables. Specifically, the algorithm runs in $poly(S,N_1,N_2^S,M)$ time.

\end{proof}

\section{Upper Bound Results}
\subsection{Proof of Theorem \ref{theorem_median_up}}
\begin{proof}
Considering the inequality that $\Single \geq \Linear$, we can confine the ratio of welfare and profit through the optimal standard contract by comparing it to the ratio of welfare and profit attained through the optimal linear contract, expressed as $\frac{\W}{\Single} \leq \frac{\W}{\Linear}$. Consequently, our focus is directed towards scrutinizing the ratio of welfare and profit achieved through the optimal linear contract, denoted as $r = \frac{\W}{\Linear}$.

A breakpoint $\alpha$ marks the threshold at which the agent's chosen action profile undergoes a transition. It's important to recall that the expected reward by an action profile $\boldsymbol{a}$ is denoted by $R_{\boldsymbol{a}}$, while the associated expected cost is represented as $c_{\boldsymbol{a}}$. For any two action profiles $\boldsymbol{a}$ and $\boldsymbol{a}'$ where $R_{\boldsymbol{a}}<R_{\boldsymbol{a}'}$, we denote $\alpha_{\boldsymbol{a},\boldsymbol{a}'}$ as the breakpoint where the agent becomes indifferent between adopting action profiles $\boldsymbol{a}$ and $\boldsymbol{a}'$. Below this threshold, the agent will favor action profile $\boldsymbol{a}$, while beyond it, she will opt for $\boldsymbol{a}'$. Specifically, the breakpoint $\alpha_{\boldsymbol{a},\boldsymbol{a}'}$ satisfies the equation:
\begin{align*}
\alpha_{\boldsymbol{a},\boldsymbol{a}'} R_{\boldsymbol{a}}-c_{\boldsymbol{a}}=\alpha_{\boldsymbol{a},\boldsymbol{a}'}R_{\boldsymbol{a}'}-c_{\boldsymbol{a}'}.
\end{align*}
Drawing a parallel to Observation 6 in \cite{dutting2019simple}, for any two action profiles $\boldsymbol{a}$ and $\boldsymbol{a}'$ where $R_{\boldsymbol{a}}< R_{\boldsymbol{a}'}$ and the breakpoint is $\alpha_{\boldsymbol{a},\boldsymbol{a}'}= \frac{c_{\boldsymbol{a}}-c_{\boldsymbol{a}'}}{R_{\boldsymbol{a}}-R_{\boldsymbol{a}'}}$, the following inequality holds:
\begin{align*}
(R_{\boldsymbol{a}'}-c_{\boldsymbol{a}'})-(R_{\boldsymbol{a}}-c_{\boldsymbol{a}})\le (1-\alpha_{\boldsymbol{a},\boldsymbol{a}'})R_{\boldsymbol{a}'}.
\end{align*}
 Sorting the breakpoints in ascending order, and  defining $\alpha_{0,1}=1, R_{\boldsymbol{a}_{0}}=c_{\boldsymbol{a}_{0}}=0$, we assume there are $n$ breakpoints. We can bound $r$ by the number of breakpoints $n$ as follows:
\begin{align*}
\W=\max_{\boldsymbol{a}}(R_{\boldsymbol{a}}-c_{\boldsymbol{a}})&=(R_{\boldsymbol{a}_n}-c{\boldsymbol{a}_n})\\&=\sum_{i=0}^{n-1}(R_{\boldsymbol{a}_{i+1}}-c_{\boldsymbol{a}_{i+1}})-(R_{\boldsymbol{a}_{i}}-c_{\boldsymbol{a}_{i}})\\&\le \sum_{i=0}^{n-1}(1-\alpha_{i,i+1})R_{\boldsymbol{a}_{i+1}}\\&\le n \Single.
\end{align*}
It's important to note that there are at most $N_1(N_2)^S$ action profiles; hence, at most $N_1(N_2)^S$ breakpoints. However, this upper bound is not tight. 

If every incentivized final action is fixed, there are at most $N_1$ breakpoints between. Specifically, considering the breakpoints on every state $s$, it's evident that there are at most $N_2$ breakpoints on each state representing different incentivized final actions. If we arrange these breakpoints across all states, there are at most $SN_2$ breakpoints. Due to the definition of the breakpoints, if $\alpha$ lies within any two neighboring breakpoints, the incentivized final action on every state remains consistent. Thus, within this range, the incentivized action profiles change only in initial actions, and each initial action is incentivized only once. This eliminates instances where initial action $i$ is incentivized at $\alpha^1$ and $\alpha^3$, while initial action $i+1$ is incentivized at $\alpha^2$, where $\alpha^1<\alpha^2<\alpha^3$. Consequently, for any two neighboring breakpoints $\alpha_i$ and $\alpha_{i+1}$, there will be at most $N_1$ new breakpoints within. This is because we can view the two-stage delegation process as a standard process with only initial actions and the intermediate state as the outcome. The expected reward of the incentivized final action on every state serves as the reward for every outcome. Consequently, there are at most $N_1$ breakpoints, implying that every initial action is incentivized once.

Based on this property, we deduce that there are at most $N_1SN_2$ breakpoints, and thus $r\le O(N_1N_2S)$. This confirms the original statement.

\end{proof}
\subsection{Proof of Claim \ref{claim_special}}
\begin{proof}
In the deterministic first-stage process, without loss of generality, let's set $n=S = N_1=\min\{S,N_1\}$. If $S < N_1$, some initial actions will dominate others, which we can straightforwardly discard. Similarly, if there are states with no initial actions leading to them, we discard those states.

To establish the original statement, we first demonstrate that the deterministic first-stage process can be reduced to a standard contract instance with $nN_2$ actions. This reduction can be performed straightforwardly by following the approach in Theorem \ref{theorem_median_exa}. Then, similar with  Theorem \ref{theorem_median_up}, we can bound the welfare-to-profit ratio under the optimal linear contract, denoted as  $r = \frac{\W}{\Linear}$. Consequently, we only need to bound the number of breakpoints as in Theorem \ref{theorem_median_up}. Given that we have already reduced the deterministic first-stage process to a standard contractual process with $nN_2$ actions, the number of breakpoints is bounded by $nN_2$. Hence, we can assert the original statement.

Combining all these insights, we can conclude the original statement.

For the stochastic first-stage process, the result follows directly based on Theorem \ref{theorem_median_up}.
\end{proof}
    \subsection{Proof of Claim \ref{theorem_tree}}\label{proof_tree}
\begin{proof}
(1) We first show that $\M=\Single$:

Indeed, it's evident that $\M \geq \Single$, as the range of pay-halfway contracts encompasses standard contracts. Therefore, we only need to establish $\Single \geq \M$.

Let the optimal pay-halfway contract be denoted by $(\boldsymbol{s}, \boldsymbol{t})$, where $s_s$ represents the payment assigned to intermediate state $s$, and $t_m$ represents the payment for outcome $m$. We aim to show that a standard contract with payment $t_m+s_s$ for outcome $m$ with predecessor state $s$ can yield equivalent profit as the optimal pay-halfway contract by demonstrating that it incentivizes the same action and incurs the same expected payment.

Initially, we establish that this standard contract incentivizes identical final actions at each intermediate state. Assuming the optimal pay-halfway contract incentivizes action $j^s$ at intermediate state $s$, the following holds:
\begin{align*}
\sum_{m=1}^M F_{j^s,m}^st_m-c_{j^s}^s\ge \sum_{m=1}^M F_{j,m}^st_m-c_j^s, \quad \forall j\in [N_2].
\end{align*}
For the standard contract, considering each final action $j$, the following holds:
\begin{align*}
\sum_{m=1}^M F^s_{j^s,m}(t_m+s_s)-c_{j^s}^s&=\sum_{m=1}^M F^s_{j^s,m}(t_m)-c_{j^s}^s+s_s \\&\ge \sum_{m=1}^M F^s_{j,m}t_m-c_j^s+s_s \\&=\sum_{m=1}^M F^s_{j,m}(t_m+s_s)-c_j^s.
\end{align*}
Thus, it incentivizes the same final action as the optimal pay-halfway contract at every intermediate state.

Next, we demonstrate that this standard contract incentivizes the same initial action $i^*$. Let $U^s(j^s)$ denote the expected utility for agent by taking action $j^s$ at intermediate state $s$, i.e., $$U^s(j^s)=\sum_{m=1}^M F^s_{j^s,m}t_m-c_{j^s}^s.$$ 

Therefore, the expected utility for each intermediate state $s$ is $U^s=U^s(j^s)+s_s$. For the standard contract, the expected utility for each intermediate state $s$ is as follow:
\begin{align*}
    U_s=\sum_{m=1}^MF^s_{j^s,m}(t_m+s_s)-c_{j^s}^s=U^s(j^s)+s_s=U^s.
\end{align*} Since the utility remains consistent, the assertion holds.

Finally, we establish that the standard contract incurs the same expected payment as the optimal pay-halfway contract. The expected payment for the optimal pay-halfway contract is given by $$T=\sum_{s=1}^SF_{i^*,s}(s_s+\sum_{m=1}^MF_{j^s,m}t_m).$$ The expected payment for the standard contract is $$T'=\sum_{s=1}^SF_{i^*,s}\sum_{m=1}^MF^s_{j^s,m}(t_m+s_s)).$$ The following equation demonstrates that they incur the same expected payment:
\begin{align*}
T&=\sum_{s=1}^SF_{i^*,s}(s_s+\sum_{m=1}^MF_{j^s,m}t_m)\\
&=\sum_{s=1}^SF_{i^*,s}\sum_{m=1}^MF_{j^s,m}(t_m+s_s))=T'.
\end{align*}

Combining these elements yields the desired result.

(2) We now proceed to show that $\T=\Single$:

Firstly, let's delve into the scenario where we denote the optimal standard contract as $\boldsymbol{t}$. In this context, it becomes apparent that the terminate-halfway contract $(\boldsymbol{t},\emptyset)$ generates the same profit as the optimal standard contract. This equivalence suggests that the expected profit from the terminate-halfway contract, denoted as $\T$, is greater than or equal to the profit from the optimal standard contract, $\Single$.

Conversely, let's suppose that the optimal terminate-halfway contract is represented by $(\boldsymbol{t}^*, \mathcal{S}^*)$. We define the outcome set $\mathcal{M}^*$ as the collection of outcomes for which the preceding state belongs to $\mathcal{S}^*$. Formally, for every outcome $m \in \mathcal{M}^*$, there exists a unique intermediate state $s \in \mathcal{S}^*$ such that $F_{s,m} \ge 0$. This formulation implies that the optimal terminate-halfway contract ensures zero payment for all outcomes within $\mathcal{M}^*$.

Utilizing this insight, we can construct a standard contract as follows: for any outcome $m \in \mathcal{M}^*$, the payment is set to be zero, while for other outcomes, the payment mirrors that specified in the optimal terminate-halfway contract. This adaptation of the standard contract guarantees at least the same performance as the optimal terminate-halfway contract, as both contracts incentivize the same action profiles with the same expected payment.

In summary, considering these observations collectively leads us to the conclusion.
\end{proof}

\subsection{Proof of Claim \ref{theorem_median_r}}
\begin{proof}
Since there's only one initial action, assigning a transfer to the intermediate state would serve no purpose in incentivizing optimal action profile. Therefore, we can straightforwardly deduce the result.

\end{proof}

\section{Construction of Instances}\label{app:proof}

\subsection{Instance of the Pay-halfway Contract with Significantly-reduced Payments}
\noindent\emph{Proof of Theorem \ref{theorem_median_cost}.}
    In this construction, we have $3$ intermediate states $s_1, s_2, s_3$, $2$ initial actions $a_1, a_2$ with zero cost, and $2$ final actions $a^s_1, a^s_2$, where $a^s_2$ represents the null final  action. The outcomes $1$ and $2$ result in rewards $0$ and $x$ ($x>0$) respectively.
\begin{itemize}
    \item In state $s_1$, both final actions have zero cost, and final action $a^s_1$ assigns probability $1$ to outcome $2$.
    \item  In state $s_2$, final action $a^s_1$ assigns probability $1$ to outcome $2$ with cost $c$ ($c< x$), and the null final action $a^s_2$ definitely leads to outcome $1$.
\item     In state $s_3$, both final actions have zero cost, leading to outcome $1$.
\end{itemize}

Initial action $a_1$ assigns probability $p$ to state $1$ and $1-p$ to state $3$, while $a_2$ assigns probability $q$ to state $1$ and $1-q$ to state $2$, where $p>q$.

Through simple observation, it becomes evident that the principal only needs to consider the final action $a_1^s$ for each state, as it is the sole action yielding a positive reward. Consequently, the principal only needs to determine which initial action to incentivize.
Upon computation, the minimum-payment pay-halfway contract to incentivize action $a_1$ assigns a payment of $0$ to every outcome, resulting in a profit of $px$ and welfare of $px$. However, to incentivize $a_2$ and $a^s_1$, the minimum-payment pay-halfway contract assigns $\frac{p-q}{1-q}c$ to state $2$ and $c$ on outcome $2$. Hence, the principal realizes a profit of $x-(1+p-q)c$ and welfare of $x-(1-q)c$. Let's assume $x-(1+p-q)c > px$, in other words,  $x > \frac{1+p-q}{1-p}c$. Under this condition, it is straightforward to deduce that $(a_2,a^s_1)$ represents the optimal action profile. Moreover, the expected payment for the optimal pay-halfway contract to incentivize the optimal action profile is then $(1+p-q)c$.

On the other hand, if a standard contract aims to incentivize the optimal action profile $(a_2,a^s_1)$,  it is straightforwardly derived that it should assign a payment of $0$ to outcome $1$. Let the payment of outcome $2$ be denoted as $r$. Then the following inequality holds:
\begin{align*}
pr &\leq rq + (1-q)(r-c), \\
r &\geq \frac{(1-q)c}{1-p}.
\end{align*}

The ratio of the payments is:
\begin{align*}
\frac{r}{(1+p-q)c} = \frac{1-q}{(1+p-q)(1-p)} \ge \frac{1-q}{1-p^2}.
\end{align*}

By letting $q$ to be a constant and then allowing $p$ to approach $1$, we can derive the desired result.

\subsection{Tight Counterexample of Pay-halfway Contract in
Deterministic First-stage Process}
\noindent\emph{Proof of Theorem \ref{theorem_median_exa}.}
    Suppose there are two outcomes with rewards $r_1=0$ and $r_2=r$ (where $r$ is large enough to satisfy later probability constraints). 
    
    There are also $N_1+1$ initial actions leading to $N_1+1$ different intermediate states. Initial action $a_{N_1+1}$ incurs a cost of $c=\sum_{k=1}^{(N_1+1)N_2+1}\lambda^k$ (where $\lambda$ is a large constant) to state $s_{N_1+1}$ while any other initial action $i$ incurs zero cost and leads to state $i$. Additionally, assume that the final action $a_{N_2+1}^s$ is a zero-cost action. 

In state $s_{N_1+1}$, every final action assigns probability $p$ to outcome $2$ with zero cost, resulting in an expected reward of $R=c+(N_1+1)N_2+1$.

Considering final action $a^s_j$ (for $j\in [N_2]$) in state $s$ where $s\ne N_1+1$, the cost will be $c_j^s=\sum_{k=1}^{sN_2+j} \lambda^k$, and the probability $p^s_j$ to outcome $2$ will ensure that the expected reward is $R^s_j=c_j^s+sN_2+j$.

If we opt for a pay-halfway contract, assigning payment $c$ to state $s_{{N_1}+1}$ suffices. This contract will incentivize initial action $a_{N_1+1}$ and then extract the total welfare $\W=R-c=(N_1+1)N_2+1$.

On the contrary, for the optimal standard contract, we can simplify the deterministic first-stage process to a standard contract instance with $(N_1+1)N_2$ actions. For each initial action that deterministically leads to one state, consider any pair of actions $(a^1, a^2)$, where $a^1$ represents an initial action and $a^2$ a final action. It assigns distribution $\mathcal{F}^s_{a^2}$ to the outcomes, where $s$ is the intermediate state that the initial action $a^1$ leads to. Consequently, the deterministic first-stage  process can be reduced to an instance of the standard contractual process with $(N_1+1)N_2$ actions. 

  For a standard contract instance, with two outcome (one with zero reward), the optimal standard contract will yields the same profit as the optimal linear contract.  To prove this, we only need to show that the optimal standard contract assigns a payment of $0$ to outcome with zero-reward. Thus, we demonstrate that the minimal payment to incentivize an optimal action 
 assigns a payment of $0$ to outcome with zero reward. Suppose the optimal standard contract assigns $y$ to the outcome with zero-reward and $x$ to the other outcome, incentivizing action $a^*$. Let $p_a$ represent the probability towards the outcome with a non-zero reward if action $a$ is taken. Then, for any other action $a$, the following inequalities hold:
\begin{align*}
    p_{a^*}x+(1-p_{a^*})y-c_{a^*}&\ge p_{a}x+(1-p_a)y-c_a\\
    (p_{a^*}-p_a)(x-y) &\ge c_{a^*}-c_a
\end{align*}

This leads to $x\ge \frac{c_{a^*}-c_a}{p_{a^*}-p_a}+y$ for action $a$ where $c_a<c_{a^*}$, and $x\le \frac{c_{a^*}-c_{a'}}{p_{a^*}-p_{a'}}+y$ for action $a'$ where $c_{a'}>c_{a^*}$. Sorting these actions by cost, we obtain $ \frac{c_{a^*}-c_a}{p_{a^*}-p_a}\le x-y\le \frac{c_{a^*}-c_{a'}}{p_{a^*}-p_{a'}}$.

According to the property of the optimal standard contract, $x$ and $y$ should be the minimal payments to incentivize such actions. Thus, $y=0$ is easily obtained.

Therefore, we narrow our focus to linear contracts to approximate the maximum profit achievable by the optimal standard contract. When considering a linear contract, the optimal fraction $\alpha$ for the optimal linear contract must correspond to a breakpoint. Thus, we  compute all the breakpoints, among which the optimal linear contract selects one. Consequently, the profit of the optimal linear contract can be expressed as:
    \begin{align*}
        \max_i{(1-\alpha_i)R_{i+1}}= \max_i (1-\frac{\lambda^{i+1}}{\lambda^{i+1}+1})(i+1+\sum_{k=1}^{i+1}\lambda^k)=\frac{i+1+\sum_{k=1}^{i+1}\lambda^k}{\lambda^{i+1}+1}\le O(1).
    \end{align*}
    
    Combining these, the ratio between the optimal pay-halfway contract and the optimal standard contract should be:
    \begin{align*}
        \M/\Single\ge \frac{(N_1+1)N_2+1}{O(1)}=\Omega(N_1N_2).
    \end{align*}

\subsection{Tight Counterexample of Terminate-halfway Contract in Deterministic First-stage Process}\label{proof_terminal_d_exa}
\noindent\emph{Proof of Theorem \ref{theorem_terminal_d_exa}.}
We analyze the same instance construction as described in Theorem \ref{theorem_median_exa}. Assume there are two outcomes with rewards $r_1 = 0$ and $r_2 = r$, where $r$ is sufficiently large to satisfy subsequent probability constraints. 

Additionally, there are $N_1 + 1$ initial actions, each leading to a distinct intermediate state among $N_1 + 1$ different intermediate states. The initial action $a_{N_1+1}$ incurs a cost of $c = \sum_{k=1}^{(N_1+1)N_2+1} \lambda^k$ (where $\lambda$ is a large constant) to reach state $s_{N_1+1}$, whereas any other initial action $i$ incurs zero cost and leads to state $i$.

In state $s_{N_1+1}$, every final action assigns a probability $p$ to outcome $2$ at zero cost, resulting in an expected reward of $R = c + (N_1+1)N_2 + 1$.

For a final action $a^s_j$ in state $s \neq s_{N_1+1}$, the cost is given by $c_j^s = \sum_{k=1}^{sN_2+j} \lambda^k$, and the probability $p^s_j$ assigned to outcome $2$ will ensure that the reward is $R^s_j = c_j^s + sN_2 + j$.

If we opt for a terminate-halfway contract, assigning a payment of $c$ to outcome $2$ and terminating the delegation process for states other than state $s_{N_1+1}$ is sufficient. This contract will incentivize the selection of initial action $a_{N_1+1}$, thereby extracting the total welfare $\W = R - c = (N_1+1)N_2 + 1$.

In contrast, for the optimal standard contract, the profit can be expressed as $O(1)$. (As shown in Theorem \ref{theorem_median_exa})

Combining these results,  we obtain the conclusion.

\subsection{Tight Counterexample of Terminate-halfway Contract in Stochastic First-stage Process}
\noindent\emph{Proof of Theorem \ref{theorem_terminal_r_example}.}
Let us consider a scenario characterized by $S+3$ potential outcomes and $2S+1$ intermediate states. The initial action has an equal probability of leading to each intermediate state. Within each state, the final action $N_2+1$ is defined as a null action, directly leading to outcome $1$.

In this scenario, the rewards for the $S+3$ outcomes are as follows:  $r_1=0, r_3=0, r_4=0,...,r_{S+3}=0$ and $r_2=r$ (where $r$ is sufficiently large). For a final action $j$ taken in state $s \in [S]$ , the associated cost is $c_j^s=\sum_{k=1}^{sN_2+j} \lambda^k$ (where $\lambda$ is a large constant), and the expected reward of this action is  $R^s_j=c_j^s+sN_2+j$.

Furthermore, the probability distribution for each final action $j$ leading to outcomes at state $s\in[S]$ is as follows:
  \begin{alignat*}
    \forall s \in [S], j \in [N_2], \quad
    &&F^{s}_{j,m} &= \begin{cases}
    1-\epsilon-\frac{R_j^s}{r}
     & \text{ if } m=1 \\
      \frac{R^s_j}{r}                                  & \text{ if } m=2 \\
      \epsilon \indicator{j = N_2}                 & \text{ if } m = s+2 \\
      \epsilon \indicator{j \not = N_2} & \text{ if } m =S+3  \\
      0 &\text{otherwise}\\
    \end{cases}.
    \end{alignat*}
Concerning the state $s\in [S+1,2S+1]$, the probability distributions are specified as follows:
    
    \begin{alignat*}
        \forall j \in [N_2], m \in [S+3], \quad
    &&F^{S+1}_{j,m} &= \begin{cases}
      1      & \text{ if } m = s-S+2 \\
     0 & \text{ otherwise}
    \end{cases}.
  \end{alignat*}
  
Let us now focus on the optimal terminate-halfway contract. In this contract, we assign a payment of  $t_{s+2}=\frac{c^s_{N_2}}{\epsilon}$ to outcome $s+2$,  and terminate the contract if the state $s$ falls within the range $s\in [S+1,2S+1]$.

Under this contract design, the agent is incentivized to select the final action $N_2$ in  states $s\in [S]$, and to terminate the delegation process if any other state is reached.  This is because choosing any final action other than $N_2$ results in a non-positive reward for the agent.  For the final action $N_2$, the agent's expected reward is $\epsilon t_{s+2} - c_{N_2}^s = 0$. Consequently, the expected profit for each state $s\in[S]$ under the terminate-halfway contract is as follows:

\begin{align*}
    \T^s=\sum_{m=1}^{S+3}F^s_{N_2,m}(r_m-t_m)=\frac{R_{N_2}^s}{r}r-\epsilon t_{s+2}= R_{N_2}^s-c_{N_2}^s=(s+1)N_2.
\end{align*}

The expected profit at state $s\in[S+1,2S+1]$ is zero, denoted as $\T^{s}=0$.
Consequently, the total profit is calculated as follows:
\begin{align*}
    \T=\frac{1}{2S+1}\sum_{s=1}^{2S+1} \T^s=\sum_{s=1}^S \frac{(s+1)N_2}{2S+1} =\Omega (SN_2).
\end{align*}

Before delving into the optimal standard contract, let's conduct a welfare analysis of the constructed scenario. For each state $s\in[S]$, the maximal welfare $\W^s$ is defined as:
\begin{align*}
    \W^s=\max_{j\in [N_2+1]}(R^s_j-c^s_j)=\max_{j\in [N_2+1]} (sN_2+j)=(s+1)N_2.
\end{align*}

The welfare at state $s\in[S+1,2S+1]$ is $0$. Therefore, the cumulative welfare is calculated as:
\begin{align*}
    \W=\frac{1}{2S+1}\sum_{s=1}^{2S+1} \W^s=\sum_{s=1}^S \frac{(s+1)N_2}{2S+1}=\frac{(S+3)SN_2}{2(2S+1)}.
\end{align*}

Now, let's shift our focus to the optimal standard contract. We start by demonstrating that, in the optimal standard contract, minimal payment is allocated to outcomes other than $1$ and $2$.

Let $\boldsymbol{t}=(t_1,t_2,t_3,...,t_{S+3})$ denote the optimal standard contract. We initially establish $t_{s}< S(2S+1)N_2$ for $s\in [3,S+3]$.

We substantiate this claim through a proof by contradiction. Assuming $t_s\ge S(2S+1)N_2$, it would lead to the highest profit $\W-\frac{1}{(2S+1)}\cdot t_s\le 0$, resulting in a contradiction.

Therefore, we only need to consider the case where at most $t_1$ and $t_2$ exceed $S(2S+1)N_2$. Moreover, by choosing $\epsilon$ to be sufficiently small, the influence of the remaining payments on the agent’s incentives becomes negligible. 

 At state $s\in[S]$, if the agent aims to incentivize action $j$, the following inequality holds:
\begin{align}
    \sum_{m=1}^{S+3}F^s_{j,m}t_m-c^s_j\ge \sum_{m=1}^{S+3}F^s_{j',m}t_m-c^s_{j'},\quad \forall j'\ne j.
    \end{align}

  By considering separate cases, the following specific inequalities are derived:
    \begin{align}
    \begin{cases}
    F^s_{j,1}t_1+F^s_{j,2}t_2-c^s_j+\epsilon t_{s+2}-\epsilon t_{S+3} \ge F^s_{j',1}t_1+F^s_{j',2}t_2-c^s_{j'}, \quad j=N_2, j'\ne j \in[N_2+1]\\
    F^s_{j,1}t_1+F^s_{j,2}t_2-c^s_j \ge F^s_{j',1}t_1+F^s_{j',2}t_2-c^s_{j'}, \quad j,j'\ne N_2, j'\ne j \in[N_2+1]\\
    F^s_{j,1}t_1+F^s_{j,2}t_2-c^s_j-\epsilon t_{s+2}+\epsilon t_{S+3} \ge F^s_{j',1}t_1+F^s_{j',2}t_2-c^s_{j'}, \quad j\ne N_2, j'=N_2
    \end{cases}.\label{eq}
\end{align}

In the context of a standard contract, we can simplify the scenario by exclusively considering the first two outcomes. This simplification is justified by the fact that the probabilities leading to outcomes $1$ and $2$ collectively amount to $1-\epsilon$. Additionally, the payments assigned to the outcomes other than $1$ and $2$ primarily function to modify the cost of the final action, with the adjustment capped at $S(2S+1)N_2\epsilon$.

  Similar to Theorem \ref{theorem_median_exa}, we show that the minimal payment required to incentivize an optimal action profile assigns $0$ to outcome with zero reward when there are only two possible outcomes. Suppose at state $s$, $p^s_a$ is the probability that final action $a$ leads to the non-zero outcome. Assume  the optimal standard contract assigns $y$ to the outcome with zero reward and $x$ to the other outcome. To incentivize the optimal action profile $\boldsymbol{a}=(a^1,a^2,...,a^S)$ where $a^s$ is the incentivized final action at state $s$, we derive the following inequalities:
\begin{align*}
    p^s_{a^s}x+(1-p^s_{a^s})y-c^s_{a^s}&\ge p^s_{a^{s'}}x+(1-p^s_{a^{s'}})y-c^s_{a^{s'}},\quad \forall s \in [S], a^{s'} \in [N_2],\\
    (p^s_{a^s}-p^s_{a^{s'}})(x-y) &\ge c^s_{a^s}-c^s_{a^{s'}}.
\end{align*}
From these, we deduce $x\ge \frac{c^s_{a^s}-c^s_{a}}{p^s_{a^s}-p^s_{a}}+y$ for action $a$ where $p^s_{a^s}>p^s_{a}$, and $x\le \frac{c^s_{a^s}-c^s_{a'}}{p^s_{a^s}-p^s_{a'}}+y$ for action $a'$ where $p^s_{a^s}<p^s_{a'}$. By sorting these actions, we obtain $ \frac{c^s_{a^s}-c^s_{a}}{p^s_{a^s}-p^s_{a}}\le x-y\le \frac{c^s_{a^s}-c^s_{a'}}{p^s_{a^s}-p^s_{a'}}$. 
According to the property of the optimal standard contract, $x$ and $y$ should be the minimal payments to incentivize such actions, thus leading to $y=0$.

Again note that we can interpret the assignment of payment to the outcomes other than
$1$ and $2$ as a modification of payment and given that the probabilities associated with outcomes $1$ and $2$ collectively sum up to $1-\epsilon$, the aforementioned inequality consistently holds true.

Taking into account the above analysis, we conclude that the optimal standard contract will assign zero payment to outcome $1$.

Therefore, our attention turns to linear contracts involving the first two outcomes to approximate the maximum profit attained by the optimal standard contract. We compute all breakpoints, from which the optimal linear contract selects one as the fraction factor. Hence, the profit of the optimal linear contract is given by:
  \begin{align*}
   \Single&=\frac{1}{2S+1}\sum_{s=1}^{2S+1}\max_i{(1-\alpha_i)R^s_{j^*}}\\ &=\frac{1}{2S+1}\max_i (1-\frac{\lambda^{i+1}}{\lambda^{i+1}+1})(i+1+\sum_{k=1}^{i+1}\lambda^k+\sum_{s=1}^{S+1} \indicator{(s+1)N_2\le i}R^s_{N_2})\\&\le \frac{i+1+\sum_{k=1}^{i+1}\lambda^k}{\lambda^{i+1}+1}\le O(1).
    \end{align*}

Assigning payment to outcomes other than $1$ and $2$ essentially amounts to modifying the cost. Moreover, we choose $\epsilon$ to be significantly small. In case $1$, where the modified cost does not affect the incentivized action, the profit increases by at most a constant due to the cost changes. In case $2$, where the modified cost changes the incentivized action, as demonstrated in the previous analysis \eqref{eq}, it will, at most, lead to a switch in the final action from $N_2$ to $N_2-1$ or from $N_2-1$ to $N_2$. However, this adjustment will only marginally impact the profit, potentially altering it by a constant factor at most.

Considering these analyses, we observe that the ratio between the terminate-halfway contract and the profit under the optimal standard contract should be at least:
    \begin{align*}
        \T/\Single\ge \frac{\Omega(SN_2)}{O(1)}=\Omega(SN_2).
    \end{align*}
Thus, we arrive at our conclusion.


\subsection{Instance Where Pay-Halfway Contract Outperforms Terminate-Halfway Contract}

\noindent\emph{Proof of Theorem \ref{obs_seperation}.}
We consider the same instance construction as in Theorem \ref{theorem_median_cost}. Specifically, in this construction, we have $3$ states $s_1, s_2, s_3$, $2$ initial actions $a_1, a_2$ with zero cost, and $2$ final actions $a^s_1, a^s_2$, where $a^s_2$ represents the null action. The outcomes $1$ and $2$ result in rewards $0$ and $x$ ($x>0$) respectively.
\begin{itemize}
    \item In state $s_1$, both final actions have zero cost, and final action $a^s_1$ assigns probability $1$ to outcome $2$.
    \item  In state $s_2$, final action $a^s_1$ assigns probability $1$ to outcome $2$ with cost $c$ ($c< x$), and null action $a^s_2$ definitely leads to outcome $1$.
\item     In state $s_3$, both actions have zero cost, leading to outcome $1$.
\end{itemize}

Initial action $a_1$ assigns probability $p$ to state $1$ and $1-p$ to state $3$, while $a_2$ assigns probability $q$ to state $1$ and $1-q$ to state $2$, where $p>q$.

Following the proof, we know that the optimal action profile is $(a_2, a_1^s)$. If we use the terminate-halfway contract to incentivize this action profile, we have two options: either "blocking" state $s_1$ or not "blocking" any state.

When considering a terminate-halfway contract that does not block any state, it becomes equivalent to a standard contract, which yields less profit than the optimal pay-halfway contract. On the other hand, if the terminate-halfway contract blocks state $s_1$, it should assign a payment of $c$ for outcome $2$, resulting in a maximal profit of $(1-q)(x-c)$.

Comparatively, the profit from the optimal pay-halfway contract is $x - (1+p-q)c$. To determine which contract yields a higher profit, we set up the inequality:  $x-(1+p-q)c>(1-q)(x-c)$ Simplifying this inequality, we get: $x>\frac{p}{q}c$. Therefore, the pay-halfway contract yields higher profit. This demonstrates the superiority of the pay-halfway contract over the terminate-halfway contract.

\end{document}